\documentclass[12pt,twoside]{article}

\usepackage[a4paper]{geometry}
\makeatletter
\renewcommand\title[1]{\gdef\@title{\reset@font\Large\bfseries #1}}
\renewcommand\section{\@startsection {section}{1}{\z@}%
                                   {-3.5ex \@plus -1ex \@minus -.2ex}%
                                   {2.3ex \@plus.2ex}%
                                   {\normalfont\large\bfseries}}
\renewcommand\subsection{\@startsection{subsection}{2}{\z@}%
                                     {-3ex\@plus -1ex \@minus -.2ex}%
                                     {1.5ex \@plus .2ex}%
                                     {\normalfont\normalsize\bfseries}}
\renewcommand\subsubsection{\@startsection{subsubsection}{3}{\z@}%
                                     {-2.5ex\@plus -1ex \@minus -.2ex}%
                                     {1.5ex \@plus .2ex}%
                                     {\normalfont\normalsize\bfseries}}

\usepackage{amsmath,amssymb,amsthm,geometry,mathtools}

\usepackage{bm} 
\usepackage{cite}

\usepackage{enumerate}

\usepackage{authblk}
\usepackage{hyperref}
\hypersetup{
    colorlinks=true,        
    linkcolor=black,        
    citecolor=black,        
    urlcolor=black,         
    pdfborder={0 0 0}        
}
\usepackage{cleveref}  

\usepackage{setspace}
\newtheorem{definition}{$\mathbf{Definition}$}
\newtheorem{theorem}{$\mathbf{Theorem}$}
\newtheorem{lemma}{$\mathbf{Lemma}$}

\newtheorem{example}{Example}

\title{MDS and NMDS Codes from the Extended Twisted Generalized  Reed-Solomon Codes}
\author[1]{Yanli Wang}
\author[1]{Yanxin Chen}
\author[1*]{Tongjiang Yan}
\date{} 
\affil[1]{China University of Petroleum(East China), Qingdao 266580, Shandong, China;
*Corresponding author. E-mail: yantoji@163.com; wyl0207111@163.com;
Contributing authors: 15581953002@163.com.}

\begin{document}

\maketitle
\begin{abstract}
This paper contributes to  maximum distance separable (MDS) and near MDS (NMDS) properties of 
the extended generalized twisted Reed-Solomon (TGRS) codes.
Firstly, a family of extended TGRS (ETGRS)  are  constructed by appending three columns to the generator matrix of original TGRS codes. 
Secondly, the necessary and sufficient conditions for these codes to be MDS or almost MDS (AMDS) codes
are derived. 
Then, by analyzing the AMDS properties of their dual codes,  
the necessary and sufffcient conditions for them to be NMDS codes are established.  
Furthermore, some examples are given to verify the main results. Finally, we determine the non-generalized
Reed-Solomon (non-GRS) characteristics of them via the Schur product method.

\end{abstract}

\noindent $\mathbf{Keywords\; }$  Extended TGRS codes $\cdot$ MDS $\cdot$ NMDS $\cdot$ Non-GRS

\section{Introduction}\par
Throughout this paper, let \( \mathbb{F}_{q} \) be the finite field with \( q \) elements, 
where $q = p^m$ and $p$ is a prime.  Let $\mathbb{F}_q^* = \mathbb{F}_q \setminus \{0\}$.
Let $\mathbb{F}_q^n$ denote the vector space of all $n$-tuples over $\mathbb{F}_q$. 
A linear code $\mathcal{C}$ with parameters $[n,k,d]$ is called a MDS code 
if it reaches the Singleton upper bound, i.e., $d = n - k + 1$. 
If $d = n - k$, then $\mathcal{C}$ is said to be an AMDS code. 
The Euclidean dual code of $\mathcal{C}$ is defined as
\[
\mathcal{C}^{\perp} = \left\{ x \in \mathbb{F}_q^n \mid \langle x, c \rangle = 0 \text{ for all } c \in \mathcal{C} \right\}.
\]
If both $\mathcal{C}$ and $\mathcal{C}^{\perp}$ are AMDS codes, then $\mathcal{C}$ is called a NMDS code.

Reed-Solomon (RS) codes, generalized RS (GRS) codes and extended GRS (EGRS) codes  are classical families of MDS codes. 
Based on GRS codes, in \cite{roth1989construction}, Roth and Lempel constructed a new class of Roth-Lempel codes by adding two columns to its generator matrix, 
then they proved that the code is a non-GRS type MDS code.
In \cite{beelen2017twisted}, inspired by Sheekey's twisted Gabidulin codes, Beelen et al. first introduced TGRS codes, 
but TGRS codes are neither  necessarily MDS codes nor GRS-type.
Affected by their work, many researches have begun to focus on  TGRS
\cite{beelen2017twisted,beelen2018structural,beelen2022twisted,liu2021construction,huang2021mds,sui2022mds,sui2022mds1,sui2023new,gu2023twisted,hu2025p,zhu2022self,zhu20241,guo2023duality,zhang2022class,zhu2021self} 
and ETGRS codes to be MDS, AMDS, NMDS codes or non-GRS type.

In 2024, Zhu et al.\cite{zhu2024+} introduced ETGRS codes by adding a column to the generator matrix of TGRS codes. 
In 2025, Li et al.\cite{li2025covering} further studied two classes of ETGRS codes obtained by adding a column and deleting a row from the generator matrix of GRS codes. 
Both works showed that the  corresponding codes are MDS, AMDS and non-GRS type.

Different from those in \cite{zhu2024+} and \cite{li2025covering}, Zhang et al. \cite{zhang2025almost} constructed another class of ETGRS codes by adding a new column to TGRS generator matrices.
Specifically, this column corresponds to the twist term with the hook at any row.
From this, they also obtained several non-GRS-type self-orthogonal MDS codes.

In \cite{ma2026mds}, 
Ma et al. introduced ETGRS codes by adding two columns to the generator matrix of TGRS codes.
They proved that the MDS codes are inequivalent to GRS codes,
and further showed that the NMDS codes are inequivalent to Roth-Lempel codes.

Zhang et al. \cite{zhang2026non} constructed another class of ETGRS codes.
More precisely, 
they introduced two extra columns corresponding to the twist term 
in the generator matrix of TGRS codes with the hook placed at an arbitrary row,
and proved that this codes are MDS or AMDS as well as non-GRS-type. 
They also determined the covering radii and deep holes of these codes.

Motivated by above works, this paper mainly focuses on a class of  defined in \Cref{def 2} by adding three additional columns to the generator matrix of TGRS codes .
The rest of this paper is organized as follows. 
In \Cref{Preliminaries}, we briefly introduce basic notations related to ETGRS codes, MDS codes, NMDS codes, extended codes and Schur Product.
In \Cref{Extended}, we verify that the ETGRS codes is an extended code of the extended TGRS code.
In \Cref{MDS} and \Cref{AMDS}, we derive the necessary and sufffcient conditions for this codes to be MDS and AMDS, respectively.
In \Cref{C^}, we give the necessary and sufficient conditions for its duals to be AMDS, then prove our codes are NMDS.
In \Cref{non}, we discuss the non-GRS properties by means of the Schur product.
Finally, \Cref{Conclusion} concludes this paper.

\section{Preliminaries}\label{Preliminaries}
 Throughout this paper, we fix the following notations for convenience.
\begin{itemize}
    \item For any finite set \( S \), \( |S| \) denotes its cardinality, i.e., the number of elements contained in \( S \).

    \item Let \(E_i(j)\) denote the matrix obtained by deleting the \(j\)-th column of \(E_i\).
    
    \item Let \(S_p(q)\) denote the matrix obtained by deleting the \(q\)-th row of \(S_p\).
    
    \item Denote $[n] = \{1, 2, \dots, n\}.$
    
    \item  For any vector \( \bm{\alpha} = (\alpha_{1}, \alpha_{2}, \ldots, \alpha_{n}) \in \mathbb{F}_{q}^{n} \), we define \(\bm{\alpha^{i}}\) =\( (\alpha_{1}^{i}, \alpha_{2}^{i}, \ldots, \alpha_{n}^{i})\) with \(i \geq 0\).
   
    \item Let \(u_i = \prod_{\substack{j=1 \\ j \neq i}}^n (\alpha_i - \alpha_j)^{-1} \) with \( i=1,2,\dots,n.\)
    
    \item  Let $\bm{\alpha} = (\alpha_1, \alpha_{2},\dots, \alpha_n) \in \mathbb{F}_q^n$ be an $n$-tuple of distinct elements of $\mathbb{F}_q.$
    Defined the $r$-th  polynomials  as 
    \[\sigma_r(I) = (-1)^r\sum_{i_1 < i_2 < \cdots < i_r \in I} \alpha_{i_1} \alpha_{i_2} \cdots \alpha_{i_r},\quad 1 \leq r \leq |I|\] 
     for any finite set $I \subseteq [n].$ For simplicity of notation, write \(\sigma_r(I)\) as $\sigma_r.$
    Defined the complete symmetric polynomial of degree $r$  as 
    \[S_r(\bm{\alpha}) = 
    \sum_{\substack{r_1+r_2+\dots+r_n=r,  r_i \geq 0}} \alpha_1^{r_1} \alpha_2^{r_2} \dots \alpha_n^{r_n} \] 
    and denote $S_r(\bm{\alpha})$  simply by $S_r$.
    
    \item Let \(\Delta_L^{(2)}=\sigma_1^2 - \sigma_2,|L| = k-2\) and \(\Delta_I^{(3)}= \sigma_1^3 + \sigma_3 - 2\sigma_2\sigma_1,|I| = k.\) 
    
    \item Let \(\Delta_J^{(4)}= \left(
    - \left( \sigma_2 \sigma_1^2 + \sigma_4 - \sigma_2^2 - \sigma_3 \sigma_1 \right)
    + \sigma_1 \left( \sigma_1^3 + \sigma_3 - 2\sigma_2 \sigma_1 \right)
    \right),|J| = k-1.\)

    \item Let  \(\Delta_J^{(5)}=
    - \left( \sigma_3\sigma_1^2 + \sigma_5 - \sigma_3\sigma_2 - \sigma_4\sigma_1 \right) 
    + \sigma_2\left( \sigma_1^3 + \sigma_3 - 2\sigma_2\sigma_1 \right)
    + \sigma_1\left( \sigma_2\sigma_1^2 + \sigma_4 - \sigma_2^2 - \sigma_3\sigma_1 \right) \\
    - \sigma_1^2\left( \sigma_1^3 + \sigma_3 - 2\sigma_2\sigma_1 \right),
    |J| = k-1.\)
     \item Let  \(\Delta_L^{(5)}=
    - \left( \sigma_3\sigma_1^2 + \sigma_5 - \sigma_3\sigma_2 - \sigma_4\sigma_1 \right) 
    + \sigma_2\left( \sigma_1^3 + \sigma_3 - 2\sigma_2\sigma_1 \right)
    + \sigma_1\left( \sigma_2\sigma_1^2 + \sigma_4 - \sigma_2^2 - \sigma_3\sigma_1 \right) \\
    - \sigma_1^2\left( \sigma_1^3 + \sigma_3 - 2\sigma_2\sigma_1 \right),
    |L| = k-2.\)
      
    \item Let \( A^T \) denote the transpose of a matrix \( A \), and \( |A| \) denote its determinant.
    
    \item  Denote the following matrices by an \( n \times n \) Vandermonde matrix \(V_n(\bm{\alpha})\), 
    an \( (n+3) \times (n+3) \) diagonal matrix diag$(\bm{v})$ and an  \( (n+2) \times (n+2) \) diagonal matrix diag$(\bm{v_1})$:
    \[V_n(\bm{\alpha}) =
        \begin{pmatrix} 
        1  & \cdots & 1 \\
        \alpha_1  & \cdots & \alpha_n \\
        \vdots  & \ddots & \vdots \\
        \alpha_1^{n-1} & \cdots & \alpha_n^{n-1} 
       \end{pmatrix},
        \operatorname{diag}(\bm{v}) =
        \begin{pmatrix}
        v_1 & & & & & \\
        & \ddots & & & & \\
        & & v_n & & & \\
        & & & 1 & & \\
        & & & & 1 & \\
        & & & & & 1
        \end{pmatrix}\] and
        \[\operatorname{diag}(\bm{v_1}) =
        \begin{pmatrix}
        v_1 & & & & \\
        & \ddots & & & \\
        & & v_n & & \\
        & & & 1 & \\
        & & & & 1 
       \end{pmatrix}. \]
       
        \end{itemize}

\subsection{ETGRS Codes}
Let the integers \( k \) and \( n \) are positive and satisfy \( 3 \leq k \leq n \leq q \).
Choose $\bm{\alpha}  \in \mathbb{F}_q^n$ 
and $\bm{v} = (v_1, v_2, \dots, v_n) \in (\mathbb{F}_q^\ast)^n$.

The $[n, k]$ $ \mathcal{GRS}$ code is defined as
\[
\mathcal{GRS}_k(\bm{\alpha}, \bm{v}) = \left\{ \left( v_1 f(\alpha_1),v_2 f(\alpha_2), \dots, v_n f(\alpha_n) \right) \mid f(x) \in \mathcal{P}_k \right\},
\]
where $\mathcal{P}_k$ denotes the set of polynomials of degree less than $k$, including the zero polynomial, in $\mathbb{F}_q[x]$.

Next, we define the linear space $\mathcal{S}_{(n,k,\eta)}$ of twisted polynomials.

\begin{definition}\label{def 1} 
Let \( {\eta}  \in \mathbb{F}_q^* \) and 
the set of  the twisted polynomials be defined as
\[
\mathcal{S}_{(n,k,{\eta})}  = \left\{ f(x) = \sum_{i=0}^{k-1} f_i x^i +  \eta f_{k-1} x^{k+2} \mid f_i \in \mathbb{F}_q \right\}.
\]
\end{definition}

Then the TGRS code is defined as 
\[ \mathcal{TGRS}_k(\bm{\alpha}, \bm{v}) = \{ (v_1 f(\alpha_1), v_2 f(\alpha_2), \cdots, v_n f(\alpha_n)) \mid f(x) \in \mathcal{S}_{(n,k,{\eta})} \}.\] 
Then we introduce the new linear code $\mathcal{C}$ extended from the TGRS code.

\begin{definition}\label{def 2}
Let \( \bm{\alpha}\), \( \bm{v} \) and \( \eta \) be defined as above.
The extended TGRS code $\mathcal{C}$ is defined as
\[
\mathcal{C} = \left\{ \left( v_1 f(\alpha_1), v_2 f(\alpha_2), \dots, v_n f(\alpha_n), f_{k-1}, f_{k-2}, f_{k-3} + \delta f_{k-1} \right) \,\bigg|\, f(x) \in \mathcal{S}_{(n,k,{\eta})} \right\}
\]
with  the generator matrix  \[G = \begin{pmatrix}
1 & 1 & \cdots & 1 & 0 & 0 & 0\\
\alpha_1 & \alpha_2 & \cdots & \alpha_n & 0 & 0 & 0 \\
\vdots & \vdots & \ddots & \vdots & \vdots & \vdots& \vdots  \\
\alpha_1^{k-4} & \alpha_2^{k-4} & \cdots & \alpha_n^{k-4} & 0 & 0 & 0 \\
\alpha_1^{k-3} & \alpha_2^{k-3} & \cdots & \alpha_n^{k-3} & 0 & 0  & 1\\
\alpha_1^{k-2} & \alpha_2^{k-2} & \cdots & \alpha_n^{k-2} & 0 & 1 & 0\\
\alpha_1^{k-1} + \eta \alpha_1^{k+2} & \alpha_2^{k-1} + \eta \alpha_2^{k+2} & \cdots & \alpha_n^{k-1} + \eta \alpha_n^{k+2} & 1 & 0  & \delta
\end{pmatrix}\operatorname{diag}(\bm{v}),\]
where \( f_{k-3} \), \( f_{k-2} \) and \( f_{k-1} \) denote the coefficients of \( x^{k-3} \), \( x^{k-2} \), and \( x^{k-1} \) in \( f(x) \), respectively.
\end{definition}

Puncturing the code $\mathcal{C}$ by deleting the last coordinate, we obtain the code $\mathcal{C}_1$
whose generator matrix is
\[G_1 = \begin{pmatrix}
1 & 1 & \cdots & 1 & 0 & 0 \\
\alpha_1 & \alpha_2 & \cdots & \alpha_n & 0 & 0 \\
\vdots & \vdots & \ddots & \vdots & \vdots & \vdots  \\
\alpha_1^{k-4} & \alpha_2^{k-4} & \cdots & \alpha_n^{k-4} & 0 & 0  \\
\alpha_1^{k-3} & \alpha_2^{k-3} & \cdots & \alpha_n^{k-3} & 0 & 0  \\
\alpha_1^{k-2} & \alpha_2^{k-2} & \cdots & \alpha_n^{k-2} & 0 & 1  \\
\alpha_1^{k-1} + \eta \alpha_1^{k+2} & \alpha_2^{k-1} + \eta \alpha_2^{k+2} & \cdots & \alpha_n^{k-1} + \eta \alpha_n^{k+2} & 1 & 0 
\end{pmatrix}\operatorname{diag}(\bm{v_1}).\]

The following \Cref{lem 1,lem 2}  are needed to prove our main theorem.
\begin{lemma}[\hspace{-0.01em}\cite{abdukhalikov2026some}, Deﬁnition 2.2\hspace{-0.01em}]\label{lem 1}
The symbols $S_r$, \( u_i \), $\bm{\alpha}  \in \mathbb{F}_q^n$ are defined as before. Then
\[
\sum_{i=1}^{n} \alpha_i^{h} u_i = 
\begin{cases}
0 ,& \text{if }  0 \leq h \leq n-2, \\
S_{h-n+1}(\bm{\alpha}) ,& \text{if }  h \geq n-1.
\end{cases} \tag{1}\label{eq1}
\]
\end{lemma}

\begin{lemma}[\hspace{-0.01em}\cite{ball2015finite,dodunekov2000near}\hspace{-0.01em}]\label{lem 2}
Let \( G \) be the generator matrix of linear code $\mathcal{C}$, then
\begin{enumerate}[(1)]
    \item $\mathcal{C}$ is MDS  if and only if any \( k \) columns of \( G \) are linearly independent;
    \item $\mathcal{C}$ is AMDS  if and only if any \( k+1 \) columns have full rank and there exist \( k \) linearly dependent columns in \( G \);
    \item $\mathcal{C}$ is NMDS  if and only if any \( k-1 \) columns of \( G \) are linearly independent, any \( k+1 \) columns of \( G \) have full rank and there exists \( k \) linearly dependent columns in \( G \).
\end{enumerate}
\end{lemma}

\subsection{Extended Codes}
In \cite{sun2024extended}, Sun et al. provided a kind description of extended codes. Let $\bm{t} = (t_1, t_2, \dots, t_n) \in \mathbb{F}_q^n$ be any nonzero vector. 
Any given $[n, k, d]$ code $\mathcal{C}$ over $\mathbb{F}_q$ can be extended into an $[n+1, k, \bar{d}]$ code $\overline{\mathcal{C}}(\bm{t})$ over $\mathbb{F}_q$ as follows:
\[
\overline{\mathcal{C}}(\bm{t}) = \left\{ (c_1, \dots, c_n, c_{n+1}) : (c_1, \dots, c_n) \in \mathcal{C}, c_{n+1} = \sum_{i=1}^n t_i c_i \right\}, \]
where $\bar{d} = d$ or $\bar{d} = d+1$.

The following lemma describes the relationship between the generator and parity-check matrices of the code and its extended code.

\begin{lemma}[\hspace{-0.01em}\cite{sun2024extended}\hspace{-0.01em}]\label{lem 3}
Let $\mathcal{C}$ be an $[n, k, d]$ linear code over $\mathbb{F}_q$ with a generator matrix $G$ and a parity-check matrix $H$, 
then the generator matrix and parity-check matrix for the extended code $\overline{\mathcal{C}}(\bm{t})$ are $(G, G\bm{t}^\mathrm{T})$ and
\[
\begin{pmatrix}
H & 0^\mathrm{T} \\
\bm{t} & -1
\end{pmatrix}
.\]
\end{lemma}

Recall the monomial equivalence relation between two linear codes of the same length.
\begin{definition}[\hspace{-0.01em}\cite{ball2015finite}\hspace{-0.01em}]
Let $\mathcal{C}_1$ and $\mathcal{C}_2$ be two linear codes of the same length over $\mathbb{F}_q$ and $G$ be a generator matrix of $\mathcal{C}_1$. Then $\mathcal{C}_1$ and $\mathcal{C}_2$ are monomially equivalent
 if there is a monomial matrix $M$ such that $G M$ is a generator matrix of $\mathcal{C}_2$.
\end{definition}

The following lemma establishes the relationship between the GRS property of an extended code and its original code.

\begin{lemma}[\hspace{-0.01em}\cite{wu2024more}, Lemma 5\hspace{-0.01em}]\label{lem 4}
If the extended code $\overline{\mathcal{C}}(\mathbf{t})$ is monomially equivalent to a \(\mathcal{GRS}\) code, then the original code $\mathcal{\mathcal{C}}$ is also monomially equivalent to a \(\mathcal{GRS}\) code.
\end{lemma}

\subsection{Schur Product}

\begin{definition}
Let $\boldsymbol{x} = (x_1, x_2, \dots, x_n)$, $\boldsymbol{y} = (y_1, y_2\dots, y_n) \in \mathbb{F}_q^n$. Then the Schur product of $\boldsymbol{x}$ and $\boldsymbol{y}$ is defined as
\[
\boldsymbol{x} \star \boldsymbol{y} = (x_1 y_1,x_2 y_2, \dots, x_n y_n).
\]
The Schur product of two codes $C_1$ and $C_2$ with length $n$ is defined as
\[
\mathcal{C}_1\star \mathcal{C}_2 = \{ \boldsymbol{c}_1\star \boldsymbol{c}_2 : \boldsymbol{c}_1 \in \mathcal{C}_1, \boldsymbol{c}_2 \in \mathcal{C}_2 \}.
\]
\end{definition}
In particular, for an $[n, k]$ linear code $\mathcal{C}$, the Schur product of $\mathcal{C}$ with itself is denoted by $\mathcal{C}^2$.

\begin{lemma}[\hspace{-0.01em}\cite{zhu2024+}, Remark 2.2\hspace{-0.01em}]\label{lem 5}
For any linear codes $\mathcal{C}_1$ and $\mathcal{C}_2$, if $\mathcal{C}_1 = \langle \boldsymbol{\beta}_1,\boldsymbol{\beta}_2, \dots, \boldsymbol{\beta}_{k_1} \rangle$ 
and $\mathcal{C}_2 = \langle \boldsymbol{\gamma}_1,\boldsymbol{\gamma}_2, \dots, \boldsymbol{\gamma}_{k_2} \rangle$ 
with $\boldsymbol{\beta}_i, \boldsymbol{\gamma}_j \in \mathbb{F}_q^n$ for $i = 1, 2, \dots, k_1$ and $j = 1, 2, \dots, k_2$, then
\[
\mathcal{C}_1 \star \mathcal{C}_2 = \langle \boldsymbol{\beta}_i \star \boldsymbol{\gamma}_j : i = 1, 2, \dots, k_1, j = 1, 2, \dots, k_2 \rangle.
\]
\end{lemma}

The Schur square of a GRS code is also a GRS code, as stated below.

\begin{lemma}[\hspace{-0.01em}\cite{marquez2013non}, Proposition 10; \cite{zhu2024+}, Proposition 2.1\hspace{-0.01em}]\label{lem 6}
If $\mathcal{C}$ is an $[n,k]_q$ GRS code with $2 \leq k \leq \frac{n+1}{2}$, then $\dim(\mathcal{C}^2) = 2k-1$.
Furthermore, for $\frac{n}{2}+1 \le k \le n$, we obtain $d\big((\mathcal{C}^\perp)^2\big) = 2k-n+2$.
\end{lemma}

\section{Extended code }\label{Extended}

\begin{theorem}\label{Th 1}
Let \( \bm{t}= (t_1,t_2 \dots, t_{n+2}) \in \mathbb{F}_q^{n+2} \) and 
\(t_i = v_i^{-1}\, u_i \alpha_i^{n+2-k}\) for \(   1 \leq i \leq n,\)
\(
t_{n+1} = \delta  - S_2 - \eta\, S_5, \)
$t_{n+2} = S_1.$
Then the extended code \[ \mathcal{C} = \overline{\mathcal{C}_1}(\bm{t}),\] 
where
\(
S_1 = \sigma_1,  S_2 = \sigma_1^2 - \sigma_2
\)
and \(S_5 =\)
\(- \left( \sigma_3 \sigma_1^2 + \sigma_5 - \sigma_3 \sigma_2 - \sigma_4 \sigma_1 \right) + \sigma_2 \left( \sigma_1^3 + \sigma_3 - 2\sigma_2 \sigma_1 \right) 
+ \sigma_1 \left( \sigma_2 \sigma_1^2 + \sigma_4 - \sigma_2^2 - \sigma_3 \sigma_1 \right)  - \sigma_1^2 \left( \sigma_1^3 + \sigma_3 - 2\sigma_2 \sigma_1 \right).\)
\end{theorem}

\begin{proof} 
By \Cref{lem 3}, we only need to prove
\[
G_1 \bm{t}^T = (0, \dots, 0, 1, 0, \delta)^T.
\]
From \eqref{eq1}, it is clear that
\[
\sum_{i=1}^{n} v_i \,t_i\, \alpha_i^r =
\begin{cases}
0, & 0 \leq r \leq k-4, \\
1, & r = k-3.
\end{cases}
\]
Next by \Cref{lem 1}, it is easy to verify that
\begin{itemize}
    \item     For $r=k-2$ \\
    \[
    \begin{aligned}
    \sum_{i=1}^{n} v_i t_i \alpha_i^{k-2} + t_{n+2}
    &= \sum_{i=1}^{n} \alpha_i^{n} u_i + t_{n+2}\\
    &= S_1 + t_{n+2}\\
    &= -\sigma_1 + t_{n+2} \\
    &= 0.
    \end{aligned}
    \]

    \item     For $r=k-1$ \\
    \[
    \begin{aligned}
    \sum_{i=1}^{n} v_i t_i \left( \alpha_i^{k-1} + \eta  \alpha_i^{k+2} \right) + t_{n+1}
    &= \sum_{i=1}^{n} v_i t_i \alpha_i^{k-1} + \eta  \sum_{i=1}^{n} v_i t_i \alpha_i^{k+2} + t_{n+1} \\
    &= \sum_{i=1}^{n} u_i \alpha_i^{n+1} + \eta  \sum_{i=1}^{n} u_i \alpha_i^{n+4} + t_{n+1} \\
    &= S_2 + \eta S_5 + t_{n+1} \\
    &= \delta.
    \end{aligned}
    \]
\end{itemize}
This completes the proof.
\end{proof}

\section{MDS property}\label{MDS}
In this section, we derive the necessary and sufficient conditions for $\mathcal{C}$ to be an MDS code.
\begin{theorem}\label{Th 2}
Suppose that \( 3 \leq k \leq n \leq q \) and \( \eta \in \mathbb{F}_q^* \).
The ETGRS code $\mathcal{C}$ is MDS if and only if the following five conditions simultaneously hold :
\begin{enumerate}[(1)]
    \setlength{\itemindent}{2em}
    \item $\eta^{-1} \neq \Delta_I^{(3)}$ for any subset $I \subseteq [n]$ with $|I| = k$.
    \item $\sigma_1(J) \neq \eta \Delta_J^{(4)}$ for any subset $J \subseteq [n]$ with $|J| = k-1$.
    \item $\sigma_2(J)  + \delta \neq \eta \left( \Delta_J^{(5)} + \sigma_1(J)  \Delta_J^{(4)} \right)$ for any subset $J \subseteq [n]$ with $|J| = k-1$.
    \item $\sigma_1(L)  \neq 0$ for any subset $L \subseteq [n]$ with $|L| = k-2$.
    \item $\eta \Delta_L^{(5)} \neq \Delta_L^{(2)} - \delta$ for any subset $L \subseteq [n]$ with $|L| = k-2$.
\end{enumerate}
\end{theorem}

\begin{proof} [$\mathbf{Proof}$]  From \Cref{lem 2}, the   $[n+3,k]$  code $\mathcal{C}$ is  an MDS code if and only if any $k \times k$ submatrix of its generator matrix must be invertible.
\begin{enumerate}[(1)] 
    \item Suppose we choose \(k\) columns from the first \(n\) columns of \(G\), then the matrix is
\[B_1 = 
\begin{pmatrix}
1 & \cdots & 1 \\
\alpha_{i_1}& \cdots& \alpha_{i_k}\\
\vdots &  & \vdots \\
\alpha_{i_1}^{k-2} & \cdots & \alpha_{i_k}^{k-2} \\
\alpha_{i_1}^{k-1} + \eta \alpha_{i_1}^{k+2} & \cdots & \alpha_{i_k}^{k-1} + \eta \alpha_{i_k}^{k+2}
\end{pmatrix}
\]

\[
\begin{aligned}
\det(B_1) 
&= |V_k(\bm{\alpha})| - \eta \left( \sigma_1^3 + \sigma_3 - 2\sigma_2\sigma_1 \right) |V_k(\bm{\alpha})| \\[6pt]
&= \left( 1 - \eta \left( \sigma_1^3 + \sigma_3 - 2\sigma_2\sigma_1 \right) \right) |V_k(\bm{\alpha})|.
\end{aligned}\]
Therefore  $\det(B_1) \neq 0 $ if and only if 
\(\eta^{-1} \neq  \sigma_1^3 + \sigma_3 - 2\sigma_2\sigma_1 \) if and only if  \[\eta^{-1} \neq \Delta_I^{(3)}\]
for any subset $I \subseteq [n]$ with $|I| = k.$

 \item Suppose we choose \(k-1\) columns from the first \(n\) columns of \(G\), together with the \((n+1)\)-th column, then the matrix is
\[\begin{pmatrix}
1 & \cdots & 1 & 0 \\
\alpha_{i_1}& \cdots& \alpha_{i_{k-1}}&0\\
\vdots &  & \vdots & \vdots &\\
\alpha_{i_1}^{k-2} & \cdots & \alpha_{i_{k-1}}^{k-2} & 0 \\
\alpha_{i_1}^{k-1} + \eta \alpha_{i_1}^{k+2} & \cdots & \alpha_{i_{k-1}}^{k-1} + \eta \alpha_{i_{k-1}}^{k+2} & 1
\end{pmatrix},
\]
it's evident that this submatrix has full rank.
 \item Suppose we choose \(k-1\) columns from the first \(n\) columns of \(G\), together with the \((n+2)\)-th column, then the matrix is
\[
B_2 = 
\begin{pmatrix}
1 & \cdots & 1 & 0 \\
\alpha_{i_1}& \cdots& \alpha_{i_{k-1}}&0\\
\vdots &  & \vdots & \vdots \\
\alpha_{i_1}^{k-3} & \cdots & \alpha_{i_{k-1}}^{k-3} & 0 \\
\alpha_{i_1}^{k-2} & \cdots & \alpha_{i_{k-1}}^{k-2} & 1 \\
\alpha_{i_1}^{k-1} + \eta \alpha_{i_1}^{k+2} & \cdots & \alpha_{i_{k-1}}^{k-1} + \eta \alpha_{i_{k-1}}^{k+2} & 0
\end{pmatrix}.
\]
we can have 
\[
\begin{aligned}
\det(B_2)
&= -
\begin{vmatrix}
1 & \cdots & 1 \\
\alpha_{i_1}& \cdots& \alpha_{i_{k-1}}\\
\vdots &  & \vdots \\
\alpha_{i_1}^{k-3} & \cdots & \alpha_{i_{k-1}}^{k-3} \\
\alpha_{i_1}^{k-1} + \eta \alpha_{i_1}^{k+2} & \cdots & \alpha_{i_{k-1}}^{k-1} + \eta \alpha_{i_{k-1}}^{k+2}
\end{vmatrix} \\
&= - 
\begin{vmatrix}
1 & \cdots & 1 \\
\alpha_{i_1}& \cdots& \alpha_{i_{k-1}}\\
\vdots &  & \vdots \\
\alpha_{i_1}^{k-3} & \cdots & \alpha_{i_{k-1}}^{k-3} \\
\alpha_{i_1}^{k-1} & \makebox[0pt][c]{$\cdots$} & \alpha_{i_{k-1}}^{k-1}
\end{vmatrix}
- \eta
\begin{vmatrix}
1 & \cdots & 1 \\
\alpha_{i_1}& \cdots& \alpha_{i_{k-1}}\\
\vdots &  & \vdots \\
\alpha_{i_1}^{k-3} & \cdots & \alpha_{i_{k-1}}^{k-3} \\
\alpha_{i_1}^{k+2} & \makebox[0pt][c]{$\cdots$} & \alpha_{i_{k-1}}^{k+2}
\end{vmatrix} \\
&= \left[\sigma_1 - \eta
\left(
- \left( \sigma_2 \sigma_1^2 + \sigma_4 - \sigma_2^2 - \sigma_3 \sigma_1 \right)
+ \sigma_1 \left( \sigma_1^3 + \sigma_3 - 2\sigma_2 \sigma_1 \right)
    \right)\right]|V_{k-1}(\bm{\alpha})|
\end{aligned}
\]
Thus \(\det(B_2) \neq 0\) if and only if 
\[
\sigma_1 \neq \eta
\left(
- \left( \sigma_2 \sigma_1^2 + \sigma_4 - \sigma_2^2 - \sigma_3 \sigma_1 \right)
+ \sigma_1 \left( \sigma_1^3 + \sigma_3 - 2\sigma_2 \sigma_1 \right)
\right)
\]
if and only if
\(
\sigma_1(J)  \neq \eta \Delta_J^{(4)}
\) for any subset $J \subseteq [n]$ with $|J| = k-1$.

\item Suppose we choose \(k-1\) columns from the first \(n\) columns of \(G\), together with the \((n+3)\)-th column, then the matrix is
\[
B_3 = 
\begin{pmatrix}
1 & \cdots & 1 & 0 \\
\alpha_{i_1}& \cdots& \alpha_{i_{k-1}}&0\\
\vdots & & \vdots & \vdots \\
\alpha_{i_1}^{k-3} & \cdots & \alpha_{i_{k-1}}^{k-3} & 1 \\
\alpha_{i_1}^{k-2} & \cdots & \alpha_{i_{k-1}}^{k-2} & 0 \\
\alpha_{i_1}^{k-1} + \eta\,\alpha_{i_1}^{k+2} & \cdots & \alpha_{i_{k-1}}^{k-1} + \eta\,\alpha_{i_{k-1}}^{k+2} & \delta
\end{pmatrix}
\]
Then expanding along the last column, we decompose  \(\det(B_3)\) into the following form:
\[
\begin{aligned}
\det(B_3)  &= 
\begin{vmatrix}
1 & \cdots & 1  \\
\alpha_{i_1}& \cdots& \alpha_{i_{k-1}}\\
\vdots & & \vdots \\
\alpha_{i_1}^{k-4} & \cdots & \alpha_{i_{k-1}}^{k-4} \\
\alpha_{i_1}^{k-2} & \cdots & \alpha_{i_{k-1}}^{k-2} \\
\alpha_{i_1}^{k-1} + \eta\,\alpha_{i_1}^{k+2} & \cdots & \alpha_{i_{k-1}}^{k-1} + \eta\,\alpha_{i_{k-1}}^{k+2}
\end{vmatrix}
+ \delta  |V_{k-1}(\bm{\alpha})| \\
&=
\begin{vmatrix}
1 & \cdots & 1  \\
\alpha_{i_1}& \cdots& \alpha_{i_{k-1}}\\
\vdots & & \vdots \\
\alpha_{i_1}^{k-4} & \cdots & \alpha_{i_{k-1}}^{k-4} \\
\alpha_{i_1}^{k-2} & \cdots & \alpha_{i_{k-1}}^{k-2} \\
\alpha_{i_1}^{k-1} & \cdots & \alpha_{i_{k-1}}^{k-1}
\end{vmatrix}
+ \eta
\begin{vmatrix}
1 & \cdots & 1  \\
\alpha_{i_1}& \cdots& \alpha_{i_{k-1}}\\
\vdots & & \vdots \\
\alpha_{i_1}^{k-4} & \cdots & \alpha_{i_{k-1}}^{k-4} \\
\alpha_{i_1}^{k-2} & \cdots & \alpha_{i_{k-1}}^{k-2} \\
\alpha_{i_1}^{k+2} & \cdots & \alpha_{i_{k-1}}^{k+2}
\end{vmatrix}
+ \delta  |V_{k-1}(\bm{\alpha})| \\
&= \left[
\sigma_2 - \eta\left(
- \left( \sigma_2^2\sigma_1 + \sigma_5 - 2\sigma_3\sigma_2 \right)
+ \sigma_1\left( \sigma_2\sigma_1^2 + \sigma_4 - \sigma_3\sigma_1 - \sigma_2^2 \right)
\right) + \delta
\right] |V_{k-1}(\bm{\alpha})| \\
\end{aligned}\]
Thus \(\det(B_3) \neq 0\) if and only if \[
\sigma_2 + \delta \neq \eta\left(
- \left( \sigma_2^2\sigma_1 + \sigma_5 - 2\sigma_3\sigma_2 \right)
+ \sigma_1\left( \sigma_2\sigma_1^2 + \sigma_4 - \sigma_3\sigma_1 - \sigma_2^2 \right)
\right) \]  if and only if 
 \(\sigma_2(J)  + \delta \neq \eta ( \Delta_J^{(5)} + \sigma_1(J)  \Delta_J^{(4)} )\)
for any subset $J \subseteq [n]$ with $|J| = k-1$.
\item Suppose we choose \(k-2\) columns from the first \(n\) columns of \(G\), together with the \((n+1)\)-th and \((n+2)\)-th columns, then the matrix is
\[
\begin{aligned}
\begin{pmatrix}
1 & \cdots & 1 & 0 & 0 \\
\alpha_{i_1}& \cdots& \alpha_{i_{k-2}}&0&0\\
\vdots & & \vdots & \vdots & \vdots \\
\alpha_{i_1}^{k-3} & \cdots & \alpha_{i_{k-2}}^{k-3} & 0 & 0 \\
\alpha_{i_1}^{k-2} & \cdots & \alpha_{i_{k-2}}^{k-2} & 0 & 1 \\
\alpha_{i_1}^{k-1} + \eta\,\alpha_{i_1}^{k+2} & \cdots & \alpha_{i_{k-2}}^{k-1} + \eta\,\alpha_{i_{k-2}}^{k+2} & 1 & 0
\end{pmatrix}
\end{aligned}
\]
Hence, the submatrix is of full rank.

\item Suppose we choose \(k-2\) columns from the first \(n\) columns of \(G\), together with the \((n+1)\)-th and \((n+3)\)-th columns, then the matrix is
\[
B_4 = 
\begin{pmatrix}
1 & \cdots & 1 & 0 & 0 \\
\alpha_{i_1}& \cdots& \alpha_{i_{k-2}}&0&0\\
\vdots & & \vdots & \vdots & \vdots \\
\alpha_{i_1}^{k-3} & \cdots & \alpha_{i_{k-2}}^{k-3} & 0 & 1 \\
\alpha_{i_1}^{k-2} & \cdots & \alpha_{i_{k-2}}^{k-2} & 0 & 0 \\
\alpha_{i_1}^{k-1} + \eta\,\alpha_{i_1}^{k+2} & \cdots & \alpha_{i_{k-2}}^{k-1} + \eta\,\alpha_{i_{k-2}}^{k+2} & 1 & \delta
\end{pmatrix} \]
\[
\det(B_4) = -
\begin{vmatrix}
0 & 1 \\
1 & \delta
\end{vmatrix}
\begin{vmatrix}
1 & \cdots & 1 \\
\alpha_{i_1}& \cdots& \alpha_{i_{k-2}}\\
\vdots &  & \vdots \\
\alpha_{i_1}^{k-4} & \cdots & \alpha_{i_{k-2}}^{k-4} \\
\alpha_{i_1}^{k-2} & \cdots & \alpha_{i_{k-2}}^{k-2}
\end{vmatrix} 
= - \sigma_1 \, |V_{k-2}(\bm{\alpha})|
\]
Thus \(\det(B_4) \neq 0\) if and only if \(\sigma_1(L)  \neq 0\) for any subset $L \subseteq [n]$ with $|L| = k-2$.

\item Suppose we choose \(k-2\) columns from the first \(n\) columns of \(G\), together with the \((n+2)\)-th and \((n+3)\)-th columns, then the matrix is
\[B_5 = 
\begin{pmatrix}
1 & \cdots & 1 & 0 & 0 \\
\alpha_{i_1}& \cdots& \alpha_{i_{k-2}}&0&0\\
\vdots & & \vdots & \vdots & \vdots \\
\alpha_{i_1}^{k-3} & \cdots & \alpha_{i_{k-2}}^{k-3} & 0 & 1 \\
\alpha_{i_1}^{k-2} & \cdots & \alpha_{i_{k-2}}^{k-2} & 1 & 0 \\
\alpha_{i_1}^{k-1} + \eta\,\alpha_{i_1}^{k+2} & \cdots & \alpha_{i_{k-2}}^{k-1} + \eta\,\alpha_{i_{k-2}}^{k+2} & 0 & \delta
\end{pmatrix}\]
\[
\begin{aligned}
\det(B_5)&=
\begin{vmatrix}
0 & 1 \\
1 & 0
\end{vmatrix}
\begin{vmatrix}
1 & \cdots & 1  \\
\alpha_{i_1}& \cdots& \alpha_{i_{k-2}}\\
\vdots & & \vdots   \\
\alpha_{i_1}^{k-4} & \cdots & \alpha_{i_{k-2}}^{k-4} \\
\alpha_{i_1}^{k-1} + \eta\,\alpha_{i_1}^{k+2} & \cdots & \alpha_{i_{k-2}}^{k-1} + \eta\,\alpha_{i_{k-2}}^{k+2}
\end{vmatrix} 
+\begin{vmatrix}
1 & 0 \\
0 & \delta
\end{vmatrix}
|V_{k-2}(\bm{\alpha})|\\
&= -
\begin{vmatrix}
1 & \cdots & 1  \\
\alpha_{i_1}& \cdots& \alpha_{i_{k-2}}\\
\vdots & & \vdots   \\
\alpha_{i_1}^{k-4} & \cdots & \alpha_{i_{k-2}}^{k-4} \\
\alpha_{i_1}^{k-1}  & \cdots & \alpha_{i_{k-2}}^{k-1} 
\end{vmatrix}
- \eta
\begin{vmatrix}
1 & \cdots & 1  \\
\alpha_{i_1}& \cdots& \alpha_{i_{k-2}}\\
\vdots & & \vdots   \\
\alpha_{i_1}^{k-4} & \cdots & \alpha_{i_{k-2}}^{k-4} \\
\alpha_{i_1}^{k+2}  & \cdots & \alpha_{i_{k-2}}^{k+2} 
\end{vmatrix}
+ \delta |V_{k-2}(\bm{\alpha})| \\
&= \Bigg(
-(\sigma_1^2 - \sigma_2)
+ \eta\Bigg[
- \left( \sigma_3\sigma_1^2 + \sigma_5 - \sigma_3\sigma_2 - \sigma_4\sigma_1 \right) 
+ \sigma_2\left( \sigma_1^3 + \sigma_3 - 2\sigma_2\sigma_1 \right)\\
&\qquad + \sigma_1\left( \sigma_2\sigma_1^2 + \sigma_4 - \sigma_2^2 - \sigma_3\sigma_1 \right) 
- \sigma_1^2\left( \sigma_1^3 + \sigma_3 - 2\sigma_2\sigma_1 \right)
\Bigg]
+ \delta
\Bigg)
|V_{k-2}(\bm{\alpha})|\\
&= \left(
- \Delta_L^{(2)} + \eta\, \Delta_L^{(5)} + \delta
\right) |V_{k-2}(\bm{\alpha})| \end{aligned}
\]
Thus \(\det(B_5) \neq 0\) if and only if 
\(\eta\, \Delta_L^{(5)} \neq \Delta_L^{(2)} - \delta\) for any subset $L \subseteq [n]$ with $|L| = k-2$.

\item Suppose we choose \(k-3\) columns from the first \(n\) columns of \(G\), together with the \((n+1)\)-th, \((n+2)\)-th and \((n+3)\)-th columns, then the matrix is
\[
\begin{aligned}
B_6 &= 
\begin{pmatrix}
1 & \cdots & 1 & 0 & 0 & 0 \\
\alpha_{i_1}& \cdots& \alpha_{i_{k-3}}&0&0&0\\
\vdots & & \vdots & \vdots & \vdots & \vdots \\
\alpha_{i_1}^{k-3} & \cdots & \alpha_{i_{k-3}}^{k-3} & 0 & 0 & 1 \\
\alpha_{i_1}^{k-2} & \cdots & \alpha_{i_{k-3}}^{k-2} & 0 & 1 & 0 \\
\alpha_{i_1}^{k-1} + \eta\,\alpha_{i_1}^{k+2} & \cdots & \alpha_{i_{k-3}}^{k-1} + \eta\,\alpha_{i_{k-3}}^{k+2} & 1 & 0 & \delta
\end{pmatrix}
\end{aligned}
\]
It follows that the submatrix has full rank.
\end{enumerate}
This complets the proof.
\end{proof}

\begin{example}
Let $q=13$, $n = 5$, $k = 3$, $\alpha = (1, 2, 5, 6, 7)$.
By \Cref{Th 2} and employing Magma, taking $(\eta, \delta) = (9, 9)$ gives that $\mathcal{C}$ is an MDS code with parameters $[8, 3, 6]$.
\end{example}

\begin{example}
Let $\xi$ be a primitive element of $\mathbb{F}_{8}$. Let $n=4$, $k=4$, $\alpha=(1,\xi^3,\xi^5,\xi^6)$.
From \Cref{Th 2} and employing Magma, all pairs $(\eta,\delta)=(\xi^t,1)$ with $1\le t\le 6$ can generate  MDS codes with identical parameters $[7,4,4]$.
\end{example}

\section{AMDS property}\label{AMDS}

\begin{theorem}\label{Th 3}
Suppose that \( 3 \leq k \leq n \leq q \) and 
\( \eta \in \mathbb{F}_q^* \), then \( \mathcal{C} \) is AMDS if and only if the following six conditions hold simultaneously:
\begin{enumerate}[(1)]
    \item For any subset \( M \subseteq [n] \) with \( |M| = k+1 \), there exists a subset \( I \subset M \) with \( |I| = k \) such that
    \[
    \eta^{-1} \neq \Delta_I^{(3)}.
    \]
    \item For any subset \( I \subseteq [n] \) with \( |I| = k \), there exists a subset \( J \subset I \) with \( |J| = k-1 \) such that
    \[
    \sigma_1(J) \neq \eta \Delta_J^{(4)}.
    \]
    \item For any subset \( I \subseteq [n] \) with \( |I| = k \), there exists a subset \( J \subset I \) with \( |J| = k-1 \) such that
    \[
    \sigma_2(J)  + \delta \neq \eta \left( \Delta_J^{(5)} + \sigma_1(J)  \Delta_J^{(4)} \right).
    \]
    \item For any subset \( J \subseteq [n] \) with \( |J| = k-1 \), there exists a subset \( L \subset J \) with \( |L| = k-2 \) such that
    \[
    \sigma_1(L) \neq 0.
    \]
    \item For any subset \( J \subseteq [n] \) with \( |J| = k-1 \), there exists a subset \( L \subset J \) with \( |L| = k-2 \) such that
    \[
    \eta \Delta_L^{(5)} \neq \Delta_L^{(2)} - \delta.
    \]
    \item One of the following conditions holds:
    \begin{enumerate} [1)]
        \item There exists a subset \( I \subset [n] \) with \( |I| = k \) such that \( \eta^{-1} = \Delta_I^{(3)} \).
        \item There exists a subset \( J \subset [n] \) with \( |J| = k-1 \) such that \( \sigma_1(J)  = \eta \Delta_J^{(4)} \).
        \item There exists a subset \( J \subset [n] \) with \( |J| = k-1 \) such that \[\sigma_2(J)  + \delta = \eta \left( \Delta_J^{(5)} + \sigma_1(J)  \Delta_J^{(4)} \right). \]
        \item There exists a subset \( L \subset [n] \) with \( |L| = k-2 \) such that \( \sigma_1(L)  = 0 \).
        \item There exists a subset \( L \subset [n] \) with \( |L| = k-2 \) such that \( \eta \Delta_L^{(5)} = \Delta_L^{(2)} - \delta \).
    \end{enumerate}
\end{enumerate}
\end{theorem}

\begin{proof} [$\mathbf{Proof}$]
\( \mathcal{C} \) is AMDS if and only if  there exist \( k \) columns of \( G \) with rank at most \( k-1 \) and any \( k+1 \) columns of \( G \) are of  full rank.

The first condition holds if and only if one of the following conditions holds:
\begin{enumerate} [1)]
    \setlength{\itemindent}{2em}
            \item There exists a subset \( I \subset [n] \) with \( |I| = k \) such that \( \eta^{-1} = \Delta_I^{(3)} \).
            \item There exists a subset \( J \subset [n] \) with \( |J| = k-1 \) such that \( \sigma_1(J)  = \eta \Delta_J^{(4)} \).
            \item There exists a subset \( J \subset [n] \) with \( |J| = k-1 \) such that \[ \sigma_2(J)  + \delta = \eta \left( \Delta_J^{(5)} + \sigma_1(J)  \Delta_J^{(4)} \right). \]
            \item There exists a subset \( L \subset [n] \) with \( |L| = k-2 \) such that \( \sigma_1(L)  = 0 \).
            \item There exists a subset \( L \subset [n] \) with \( |L| = k-2 \) such that \( \eta \Delta_L^{(5)} = \Delta_L^{(2)} - \delta \).
    \end{enumerate}
Next, we consider the second condition that any \(k+1\) columns of \(G\) are of full rank 
if and only if the rank of every \(k \times (k+1)\) submatrix of \(G\) is exactly \(k\). 
For any \(k \times (k+1)\) submatrix \(E_i\) of \(G\), we have \(\operatorname{rank}(E_i) = k\) 
if and only if there exists a matrix \( E_i{(j)} \) such that \( \det(E_i{(j)}) \neq 0 .\)
We discuss this by the following cases:

\begin{enumerate}[(1)]
\item Suppose we select \( k+1 \) columns from the first \( n \) columns of \( G \), then the corresponding \( k \times (k+1) \) submatrix is
\[
E_1 = \begin{pmatrix}
1 & \cdots & 1 \\
\alpha_{i_1} & \cdots &\alpha_{i_{k+1}} \\
\vdots & & \vdots \\
\alpha_{i_1}^{k-1} + \eta \alpha_{i_1}^{k+2} & \cdots & \alpha_{i_{k+1}}^{k-1} + \eta \alpha_{i_{k+1}}^{k+2}
\end{pmatrix},
\]
where \( M = \{i_1,i_2, \dots, i_{k+1}\} \subseteq [n] \) and \( 1 \leq j \leq k+1.\)

Without loss of generality, assume \( j = k+1 \), note that \( B_1 = E_1{( k+1)} \). By \Cref{Th 2},
\(
\det(E_1{( k+1)}) = \det(B_1) \neq 0
\)
if and only if \( \eta^{-1} \neq \Delta_I^{(3)} \).

Thus, \( \operatorname{rank}(E_1) = k \) if and only if for any subset \( M \subseteq [n] \), 
then there exists a subset \( I \subset M \) with \( |I| = k \) such that
\(    \eta^{-1} \neq \Delta_I^{(3)}.\)

\item 
 Suppose we take \(k\) columns from the first \(n\) columns of \(G\) together with the \((n+1)\)-th column. Then the corresponding \(k \times (k+1)\) submatrix is
\[
\begin{pmatrix}
1 & \cdots & 1 & 0 \\
\alpha_{i_1} & \cdots &\alpha_{i_{k}} & 0\\
\vdots & & \vdots & \vdots \\
\alpha_{i_1}^{k-2} & \cdots & \alpha_{i_k}^{k-2} & 0 \\
\alpha_{i_1}^{k-1} + \eta \alpha_{i_1}^{k+2} & \cdots & \alpha_{i_k}^{k-1} + \eta \alpha_{i_k}^{k+2} & 1
\end{pmatrix},
\]
where \(\{i_1, i_2,\dots,i_k\} \subseteq [n]\).
It is easy to verify that this submatrix has full rank.

\item Suppose we take \(k\) columns from the first \(n\) columns of \(G\) together with the \((n+2)\)-th column. Then the corresponding \(k \times (k+1)\) submatrix is
\[
E_2 = 
\begin{pmatrix}
1 & \cdots & 1 & 0 \\
\alpha_{i_1} & \cdots &\alpha_{i_{k}} & 0\\
\vdots & & \vdots & \vdots \\
\alpha_{i_1}^{k-2} & \cdots & \alpha_{i_k}^{k-2} & 1 \\
\alpha_{i_1}^{k-1} + \eta \alpha_{i_1}^{k+2} & \cdots & \alpha_{i_k}^{k-1} + \eta \alpha_{i_k}^{k+2} & 0
\end{pmatrix},
\]
where \(I = \{i_1, i_2, \dots, i_k\} \subseteq [n]\) and \(1 \leq j \leq k.\)

Without loss of generality, assume \(j = k\), note that \(B_2 = E_2{ (k)}\). By \Cref{Th 2},
\(
\det(E_2{ (k)}) = \det(B_2) \neq 0
\)
if and only if \(\sigma_1 \neq \eta \Delta_J^{(4)}\).

Therefore \(\operatorname{rank}(E_2) = k\) if and only if for any subset \(I \subseteq [n]\), 
there exists a subset \(J \subset I\) with \(|J| = k-1\) such that \(\sigma_1(J) \neq \eta \Delta_J^{(4)}\).

\item Suppose we take \(k\) columns from the first \(n\) columns of \(G\) together with the \((n+3)\)-th column. Then the corresponding \(k \times (k+1)\) submatrix is
\[
E_3 = 
\begin{pmatrix}
1 & \cdots & 1 & 0 \\
\alpha_{i_1} & \cdots &\alpha_{i_{k}} & 0\\
\vdots & & \vdots & \vdots \\
\alpha_{i_1}^{k-3} & \cdots & \alpha_{i_k}^{k-3} & 1 \\
\alpha_{i_1}^{k-2} & \cdots & \alpha_{i_k}^{k-2} & 0 \\
\alpha_{i_1}^{k-1} + \eta \alpha_{i_1}^{k+2} & \cdots & \alpha_{i_k}^{k-1} + \eta \alpha_{i_k}^{k+2} & \delta
\end{pmatrix},
\]
where \(I = \{i_1, i_2,\dots, i_k\} \subseteq [n]\) and \(1 \leq j \leq k.\)

Without loss of generality, assume \(j = k\). Note that \(E_3{(k)} = B_3\). By \Cref{Th 2},
\(
\det(E_3{(k)}) = \det(B_3) \neq 0
\)
if and only if \(\sigma_2 + \delta \neq \eta \left( \Delta_J^{(5)} + \sigma_1 \Delta_J^{(4)} \right)\).

So \(\operatorname{rank}(E_3) = k\) if and only if for any subset \(I \subseteq [n]\), there exists a subset \(J \subset I\) with \(|J| = k-1\) such that
\(
\sigma_2(J) + \delta \neq \eta \left( \Delta_J^{(5)} + \sigma_1(J) \Delta_J^{(4)} \right).
\)

\item 
Suppose we take \(k-1\) columns from the first \(n\) columns of \(G\), together with the \((n+1)\)-th and \((n+2)\)-th columns. Then the corresponding \(k \times (k+1)\) submatrix is
\[
\begin{pmatrix}
1 & \cdots & 1 & 0 & 0 \\
\alpha_{i_1} & \cdots &\alpha_{i_{k-1}} & 0& 0\\
\vdots & & \vdots & \vdots & \vdots \\
\alpha_{i_1}^{k-3} & \cdots & \alpha_{i_{k-1}}^{k-3} & 0 & 0 \\
\alpha_{i_1}^{k-2} & \cdots & \alpha_{i_{k-1}}^{k-2} & 0 & 1 \\
\alpha_{i_1}^{k-1} + \eta \alpha_{i_1}^{k+2} & \cdots & \alpha_{i_{k-1}}^{k-1} + \eta \alpha_{i_{k-1}}^{k+2} & 1 & 0
\end{pmatrix},
\]
where \(J = \{i_1, i_2, \dots, i_{k-1}\} \subseteq [n]\).
Obviously, this submatrix has full rank.

\item
Suppose we take \(k-1\) columns from the first \(n\) columns of \(G\), together with the \((n+1)\)-th and \((n+3)\)-th columns. Then the corresponding \(k \times (k+1)\) submatrix is
\[
E_4 = 
\begin{pmatrix}
1 & \cdots & 1 & 0 & 0 \\
\alpha_{i_1} & \cdots &\alpha_{i_{k-1}} & 0& 0\\
\vdots & & \vdots & \vdots & \vdots \\
\alpha_{i_1}^{k-3} & \cdots & \alpha_{i_{k-1}}^{k-3} & 0 & 1 \\
\alpha_{i_1}^{k-2} & \cdots & \alpha_{i_{k-1}}^{k-2} & 0 & 0 \\
\alpha_{i_1}^{k-1} + \eta \alpha_{i_1}^{k+2} & \cdots & \alpha_{i_{k-1}}^{k-1} + \eta \alpha_{i_{k-1}}^{k+2} & 1 & \delta
\end{pmatrix},
\]
where \(J = \{i_1, i_2, \dots, i_{k-1}\} \subseteq [n]\) and \(1 \leq j \leq k-1.\)

Without loss of generality, assume \(j = k-1\). Note that \(E_4{(k-1)} = B_4\). By Theorem,
\(
\det(E_4{(k-1)}) = \det(B_4) \neq 0
\)
if and only if \(\sigma_1 \neq 0\).

Then \(\operatorname{rank}(E_4) = k\) if and only if for any subset \(J \subseteq [n]\), 
there exists a subset \(L \subset J\) with \(|L| = k-2\) such that \(\sigma_1(L) \neq 0\).

\item  
 Suppose we take \(k-1\) columns from the first \(n\) columns of \(G\), together with the \((n+2)\)-th and \((n+3)\)-th columns. Then the corresponding \(k \times (k+1)\) submatrix is
\[
E_5 = 
\begin{pmatrix}
1 & \cdots & 1 & 0 & 0 \\
\alpha_{i_1} & \cdots &\alpha_{i_{k-1}} & 0& 0\\
\vdots & & \vdots & \vdots & \vdots \\
\alpha_{i_1}^{k-3} & \cdots & \alpha_{i_{k-1}}^{k-3} & 0 & 1 \\
\alpha_{i_1}^{k-2} & \cdots & \alpha_{i_{k-1}}^{k-2} & 1 & 0 \\
\alpha_{i_1}^{k-1} + \eta \alpha_{i_1}^{k+2} & \cdots & \alpha_{i_{k-1}}^{k-1} + \eta \alpha_{i_{k-1}}^{k+2} & 0 & \delta
\end{pmatrix},
\]
where \(J = \{i_1, i_2, \dots, i_{k-1}\} \subseteq [n]\) and  \(1 \leq j \leq k-1.\)

Without loss of generality, assume \(j = k-1\). Note that \(E_5{(k-1)} = B_5\). By Theorem,
\(
\det(E_5{(k-1)}) = \det(B_5) \neq 0
\)
if and only if \(\eta \Delta_L^{(5)} \neq \Delta_L^{(2)} - \delta\).

Hence \(\operatorname{rank}(E_5) = k\) if and only if for any subset \(J \subseteq [n]\), there exists a subset \(L \subset J\) with \(|L| = k-2\) such that
\(
\eta \Delta_L^{(5)} \neq \Delta_L^{(2)} - \delta.
\)

\item 
 Suppose we take \(k-2\) columns from the first \(n\) columns of \(G\), together with the \((n+1)\)-th, \((n+2)\)-th and \((n+3)\)-th columns. Then the corresponding \(k \times (k+1)\) submatrix is
\[
\begin{pmatrix}
1 & \cdots & 1 & 0 & 0 & 0 \\
\alpha_{i_1} & \cdots &\alpha_{i_{k-2}} & 0& 0& 0\\
\vdots & & \vdots & \vdots & \vdots & \vdots \\
\alpha_{i_1}^{k-3} & \cdots & \alpha_{i_{k-2}}^{k-3} & 0 & 0 & 1 \\
\alpha_{i_1}^{k-2} & \cdots & \alpha_{i_{k-2}}^{k-2} & 0 & 1 & 0 \\
\alpha_{i_1}^{k-1} + \eta \alpha_{i_1}^{k+2} & \cdots & \alpha_{i_{k-2}}^{k-1} + \eta \alpha_{i_{k-2}}^{k+2} & 1 & 0 & \delta
\end{pmatrix},
\]
where \(\{i_1, i_2, \dots, i_{k-2}\} \subseteq [n]\).
Obviously, this submatrix has full rank.
\end{enumerate}

This completes the proof.
\end{proof}

\begin{example}
Let $\xi$ be a primitive element of $\mathbb{F}_8$. Let $n = 5$, $k = 3$, $\alpha = (1, \xi, \xi^2, \xi^4, \xi^5)$.
Taking $(\eta, \delta) = (\xi^2, 0)$ gives that $\mathcal{C}$ is an AMDS code with parameters $[8, 3, 5]$.
Using Magma, we obtain many other $(\eta, \delta)$ pairs including 
$(\xi^2,1),(\xi^2,\xi^2),(\xi^2,\xi^3),(\xi^2,\xi^4),(\xi^2,\xi^5)$, 
$(\xi^4,0),(\xi^4,\xi),(\xi^4,\xi^3),(\xi^4,\xi^4),(\xi^4,\xi^6)$, 
$(\xi^5,0),(\xi^5,1),(\xi^5,\xi),(\xi^5,\xi^3),(\xi^5,\xi^4)$, 
$(\xi^6,0),(\xi^6,1)$,
$(\xi^6,\xi^2),(\xi^6,\xi^4),(\xi^6,\xi^5),(\xi^6,\xi^6) $
that yield the same AMDS code parameters.
\end{example}

\section{ The AMDS of their dual codes}\label{C^}
\begin{theorem}\label{Th 4}
Suppose that \(3 \leq k \leq n \leq q\) and \(\eta \in \mathbb{F}_q^*\).
Then \(\mathcal{C}^\perp\) is AMDS if and only if one of the following conditions holds:
\begin{enumerate}[(1)]
    \item There exists a subset \(I \subseteq [n]\) with \(|I| = k\) such that
   \(\eta^{-1} = \Delta_I^{(3)}. \)
    \item There exists a subset \(J \subseteq [n]\) with \(|J| = k-1\) such that
   \(\sigma_1(J) = \eta \Delta_J^{(4)}.\)
    \item There exists a subset \(J \subseteq [n]\) with \(|J| = k-1\) such that
   \[\sigma_2(J)+ \delta = \eta \left( \Delta_J^{(5)} + \sigma_1(J) \Delta_J^{(4)} \right). \]
    \item There exists a subset \(L \subseteq [n]\) with \(|L| = k-2\) such that
   \(\sigma_1(L) = 0.\)
    \item There exists a subset \(L \subseteq [n]\) with \(|L| = k-2\) such that
   \(\eta \Delta_L^{(5)} = \Delta_L^{(2)} - \delta.\)
\end{enumerate}
\end{theorem}

\begin{proof} [$\mathbf{Proof}$]
The code \(C^\perp\) is AMDS if and only if there exist \(k\) columns of \(G\) with rank at most \(k-1\), and any \(k-1\) columns of \(G\) are of full rank.

The first condition holds if and only if the following conditions hold:
\begin{enumerate}[(1)]
    \setlength{\itemindent}{2em}
    \item There exists a subset \(I \subseteq [n]\) with \(|I| = k\) such that
   \(\eta^{-1} = \Delta_I^{(3)}. \)
    \item There exists a subset \(J \subseteq [n]\) with \(|J| = k-1\) such that
   \(\sigma_1(J) = \eta \Delta_J^{(4)}.\)
    \item There exists a subset \(J \subseteq [n]\) with \(|J| = k-1\) such that
   \[\sigma_2(J)+ \delta = \eta \left( \Delta_J^{(5)} + \sigma_1(J) \Delta_J^{(4)} \right). \]
    \item There exists a subset \(L \subseteq [n]\) with \(|L| = k-2\) such that
   \(\sigma_1(L) = 0.\)
    \item There exists a subset \(L \subseteq [n]\) with \(|L| = k-2\) such that
   \(\eta \Delta_L^{(5)} = \Delta_L^{(2)} - \delta.\)
\end{enumerate}

Next, we consider the second condition that  
any \(k-1\) columns of \(G\) are of full rank if and only if every \(k \times (k-1)\) submatrix of \(G\) has rank exactly \(k-1\). 
Without loss of generality, assume \(v=1\). Then we discuss this by the following cases:

\begin{enumerate}[(1)]
\item Suppose we choose \(k-1\) columns from the first \(n+3\) columns of \(G\). It is easy to verify that the chosen matrix has full rank.

\item  Suppose we take \(k-3\) columns from the first \(n\) columns of \(G\), together with the \((n+1)\)-th and \((n+2)\)-th columns. Then the corresponding \(k \times (k-1)\) submatrix is
\[
S_3 = 
\begin{pmatrix}
1 & \cdots & 1 & 0 & 0 \\
\alpha_{i_1} & \cdots &\alpha_{i_{k-3}} & 0& 0\\
\vdots & & \vdots & \vdots & \vdots \\
\alpha_{i_1}^{k-2} & \cdots & \alpha_{i_{k-3}}^{k-2} & 0 & 1 \\
\alpha_{i_1}^{k-1} + \eta \alpha_{i_1}^{k+2} & \cdots & \alpha_{i_{k-3}}^{k-1} + \eta \alpha_{i_{k-3}}^{k+2} & 1 & 0
\end{pmatrix}.
\]
It is evident that the submatrix \(S_3\) has full rank.

\item Suppose we take \(k-3\) columns from the first \(n\) columns of \(G\), together with the \((n+1)\)-th and \((n+3)\)-th columns. Then the corresponding \(k \times (k-1)\) submatrix is
\[
S_4 = 
\begin{pmatrix}
1 & \cdots & 1 & 0 & 0 \\
\alpha_{i_1} & \cdots &\alpha_{i_{k-3}} & 0& 0\\
\vdots & & \vdots & \vdots & \vdots \\
\alpha_{i_1}^{k-3} & \cdots & \alpha_{i_{k-3}}^{k-3} & 0 & 1 \\
\alpha_{i_1}^{k-2} & \cdots & \alpha_{i_{k-3}}^{k-2} & 0 & 0 \\
\alpha_{i_1}^{k-1} + \eta \alpha_{i_1}^{k+2} & \cdots & \alpha_{i_{k-3}}^{k-1} + \eta \alpha_{i_{k-3}}^{k+2} & 1 & \delta
\end{pmatrix}.
\]
Then \(\operatorname{rank}(S_4) = k-1\) if and only if there exists a matrix \(S_4(q)\) such that \(\det(S_4(q)) \neq 0\), where \(1 \leq q \leq k\).
\par
Without loss of generality, assume \(q = k-1\). 
Note that the \(k-1\) rows of \(S_4(k-1)\) are linearly independent, 
so \(\operatorname{rank}(S_4) = k-1\).

\item  Suppose we take \(k-3\) columns from the first \(n\) columns of \(G\), together with the \((n+2)\)-th and \((n+3)\)-th columns. Then the corresponding \(k \times (k-1)\) submatrix is
\[
S_5 = 
\begin{pmatrix}
1 & \cdots & 1 & 0 & 0 \\
\alpha_{i_1} & \cdots &\alpha_{i_{k-3}} & 0& 0\\
\vdots & & \vdots & \vdots & \vdots \\
\alpha_{i_1}^{k-3} & \cdots & \alpha_{i_{k-3}}^{k-3} & 0 & 1 \\
\alpha_{i_1}^{k-2} & \cdots & \alpha_{i_{k-3}}^{k-2} & 1 & 0 \\
\alpha_{i_1}^{k-1} + \eta \alpha_{i_1}^{k+2} & \cdots & \alpha_{i_{k-3}}^{k-1} + \eta \alpha_{i_{k-3}}^{k+2} & 0 & \delta
\end{pmatrix}.
\]
Then \(\operatorname{rank}(S_5) = k-1\) if and only if there exists a matrix \(S_5(q)\) such that \(\det(S_5(q)) \neq 0\), where \(1 \leq q \leq k\).
Without loss of generality, assume \(q = k\). 
Note that the first \(k-1\) rows of \(S_5(q)\) are linearly independent, 
so \(\operatorname{rank}(S_5) = k-1\).

\item Suppose we take \(k-4\) columns from the first \(n\) columns of \(G\), together with the \((n+1)\)-th, \((n+2)\)-th and \((n+3)\)-th columns. Then the corresponding \(k \times (k-1)\) submatrix is
\[
S_6 = 
\begin{pmatrix}
1 & \cdots & 1 & 0 & 0 & 0 \\
\alpha_{i_1} & \cdots &\alpha_{i_{k-4}} & 0& 0& 0\\
\vdots & & \vdots & \vdots & \vdots & \vdots \\
\alpha_{i_1}^{k-3} & \cdots & \alpha_{i_{k-4}}^{k-3} & 0 & 0 & 1 \\
\alpha_{i_1}^{k-2} & \cdots & \alpha_{i_{k-4}}^{k-2} & 0 & 1 & 0 \\
\alpha_{i_1}^{k-1} + \eta \alpha_{i_1}^{k+2} & \cdots & \alpha_{i_{k-4}}^{k-1} + \eta \alpha_{i_{k-4}}^{k+2} & 1 & 0 & \delta
\end{pmatrix}.
\]It is clear that the rank of \(S_6\) is \(k-1\).
\end{enumerate}
This completes the proof.
\end{proof}

\textbf{Example 4.5}
Let $\mathbb{F}_{11}$ be the finite field with $11$ elements. Let $n = 5$, $k = 3$, $\alpha = (0, 4, 5, 8, 9)$.
For any $\eta \in \mathbb{F}_{11}^*$ and arbitrary $\delta \in \mathbb{F}_{11}$, the dual code of $\mathcal{C}$ is an AMDS code with parameters $[8, 4, 4]$.
All these results are obtained by utilizing Magma.

\begin{theorem}
The code \(C\) is NMDS if and only if \(C\) is AMDS.
\end{theorem}

\begin{proof} [$\mathbf{Proof}$]
Combining \Cref{Th 3,Th 4}, it is straightforward to obtain the necessary and sufficient condition for \(C\) to be NMDS.
\end{proof}

\section{Non-GRS type}\label{non}
Throughout this paper, a code is called \textit{non-GRS type} if it is not monomially equivalent to a GRS code.

\begin{theorem}\label{Th 6}
For \(3 \leq k < \dfrac{n+3}{2}\), the code \(C_1\) is of non-GRS type.
\end{theorem}

\begin{proof}[$\mathbf{Proof}$]
Without loss of generality, assume \(v=1\). For \(3 \leq k < \dfrac{n+3}{2}\), 
by \Cref{lem 6}, we only need to prove that if  the dimension of the Schur product \(C_1^2\) is \(2k-1\).

By definition of \(C_1\), for \(0 \leq i \leq k-4,\)
\[
C_1 = \left\langle
(\alpha^i, 0, 0, 0),
(\alpha^{k-3}, 0, 0, 1),
(\alpha^{k-2}, 0, 1, 0),
(\alpha^{k-1} + \eta \alpha^{k+2}, 1, 0, \delta)
\right\rangle.
\quad 
\]

By \Cref{lem 5}, for \( 0 \leq i,j \leq k-4,\)
\[
\begin{aligned}
C_1^{2} = \Big\langle
& (\alpha^{i+j}, 0, 0, 0), 
(\alpha^{i+k-3}, 0, 0, 0),
(\alpha^{i+k-2}, 0, 0, 0), 
\alpha^i (\alpha^{k-1} + \eta\alpha^{k+2}, 0, 0, 0), 
(\alpha^{k-3}, 0, 0, 1)^{2},\\
&(\alpha^{k-3}, 0, 0, 1) \star (\alpha^{k-2}, 0, 1, 0), 
(\alpha^{k-3}, 0, 0, 1) \star (\alpha^{k-1} + \eta\alpha^{k+2}, 1, 0, \delta),
(\alpha^{k-2}, 0, 1, 0)^{2}, \\ 
&(\alpha^{k-2}, 0, 1, 0) \star (\alpha^{k-1} + \eta\alpha^{k+2}, 1, 0, \delta),
(\alpha^{k-1} + \eta\alpha^{k+2}, 1, 0, \delta)^{2}
\Big\rangle.
\end{aligned}
\]

For \(k \ge 3\), we have \(2k-4 \ge k-1\) and \(2k-3 \ge k\).
\[
\begin{aligned}
C_1^{2} &= 
\Big\langle
(\alpha^s, 0, 0, 0),
(\eta\alpha^{2k-1}, 0, 0, \delta),
(\eta\alpha^{2k}, 0, 0, 0),
(\eta^2\alpha^{2k+4} + 2\eta\alpha^{2k+1}, 1, 0, \delta^2),
0 \le s \le 2k-2
\Big\rangle \\
&\quad
+ \Big\langle (0, \dots, 0, 0, 0, 1)\Big\rangle
+ \Big\langle (0, \dots, 0, 0, 1, 0)\Big\rangle \\
&= 
\Big\langle
(\alpha^s, 0, 0, 0), 
(\eta^2\alpha^{2k+4} + 2\eta\alpha^{2k+1}, 1, 0, 0), 
0 \le s \le 2k
\Big\rangle 
+ \Big\langle (0, \dots, 0, 0, 0, 1)\Big\rangle\\
&\quad+ \Big\langle (0, \dots, 0, 0, 1, 0)\Big\rangle. \\
\end{aligned}
\]

For \(3 \le k < \dfrac{n+3}{2}\), we obtain the \((2k+3) \times (2k+3)\) matrix
\[
\left(
\begin{array}{ccccccc}
1 & \cdots & 1 & 0 & 0 & 0 \\
\alpha_1 & \cdots & \alpha_{n} & 0 & 0 & 0 \\
\vdots & & \vdots & \vdots & \vdots & \vdots \\
\alpha_1^{2k} & \cdots & \alpha_n^{2k} & 0 & 0 & 0 \\
\eta^2\alpha_1^{2k+4}+2\eta\alpha_1^{2k+1} & \cdots & \eta^2\alpha_n^{2k+4}+2\eta\alpha_n^{2k+1} & 1 & 0 & 0 \\
0 & \cdots & 0 & 0 & 1 & 0 \\
0 & \cdots & 0 & 0 & 0 & 1
\end{array}
\right)
\]
has a \(2k \times 2k\) invertible submatrix
\begin{align*}
N &= 
\begin{pmatrix}
1 & \cdots & 1 & 0 & 0 & 0 \\
\alpha_1 & \cdots & \alpha_{2k-3} & 0 & 0 & 0 \\
\vdots & & \vdots & \vdots & \vdots & \vdots \\
\alpha_1^{2k-4} & \cdots & \alpha_{2k-3}^{2k-4} & 0 & 0 & 0 \\
\eta^2\alpha_1^{2k+4}+2\eta\alpha_1^{2k+1} & \cdots & \eta^2\alpha_{2k-3}^{2k+4}+2\eta\alpha_{2k-3}^{2k+1} & 1 & 0 & 0 \\
0 & \cdots & 0 & 0 & 1 & 0 \\
0 & \cdots & 0 & 0 & 0 & 1
\end{pmatrix},\\
|N| &=
\begin{vmatrix}
1 & \cdots & 1 \\
\alpha_1 & \cdots & \alpha_{2k-3}\\
\vdots & & \vdots \\
\alpha_1^{2k-4} & \cdots & \alpha_{2k-3}^{2k-4}
\end{vmatrix}
= \prod_{1 \le i < j \le 2k-3} (\alpha_j - \alpha_i) \neq 0.
\end{align*}
We have \(\dim(C_1^{ 2}) \ge 2k > 2k-1\). Thus \(C_1\) is a non-GRS type code for \(3 \le k < \dfrac{n+3}{2}\).

This completes the proof.
\end{proof}

\begin{theorem}\label{Th 7}
For \(3 \le k \le n-4\), the code \(C\) is of non-GRS type.
\end{theorem}

\begin{proof}[$\mathbf{Proof}$]
Without loss of generality, assume \(v=1\). We proceed by two cases.

- $\mathbf{Case\ 1:}$ \(3 \le k < \dfrac{n+4}{2}\). 
From \Cref{Th 1}, we know that \(\mathcal{C} = \overline{\mathcal{C}_1}(\bm{t})\).
Therefore, by \Cref{Th 6}, it follows that \(C\) is a non-GRS type code for \(3 \le k < \dfrac{n+4}{2}\).

- $\mathbf{Case\ 2:}$\(\dfrac{n+4}{2} \le k \le n-4\).
 By \Cref{lem 6}, we only need to prove \(C\) is a non-GRS type code when the minimum distance of the Schur product \((C^\perp)^2\) is \(2k-n+1\).

We have \(G \bm{c_i}^T = 0\) for \(i=1,2,3\), where
\[
\begin{aligned}
\bm{c_1} &= \left( u_1 \alpha_1^{n-k-4}, u_2 \alpha_2^{n-k-4}, \dots, u_n \alpha_n^{n-k-4},0,0,0 \right), \\
\bm{c_2} &= \left( u_1 \alpha_1^{n-k-3}, u_2 \alpha_2^{n-k-3}, \dots, u_n \alpha_n^{n-k-3},-\eta,0,0 \right), \\
\bm{c_3} &= ( u_1 \alpha_1^{n-k-2}, u_2 \alpha_2^{n-k-2}, \dots, u_n \alpha_n^{n-k-2},-\eta \sum_{i=1}^n \alpha_i,0, 0 ).
\end{aligned}
\]
Thus  \(\bm{c_1},\bm{ c_2},\bm{ c_3} \in \bm{C^\perp}\) and
\(
\bm{c} =\bm{ c_1} \bm{c_3 }- \bm{c_2}^2 =(0, \dots, 0, -\eta^2, 0, 0 ) \in (\bm{C^\perp})^{2}.
\) So \(d\left((\bm{C}^\perp)^{2}\right) = 1\).

However, for an \([n+3, k]\) GRS code, we have 
\(
d\left((\text{GRS}^\perp)^{ 2}\right) = 2k - n - 1.
\)
When \(\dfrac{n+4}{2} \le k \le n-4\), we have
\(
d\left((\bm{C}^\perp)^{ 2}\right) \ge 3 \neq 1.
\)
Thus \(\bm{C}\) is of non-GRS type for \(\dfrac{n+4}{2} \le k \le n-4\).

Combining both cases, \(\bm{C}\) is a non-GRS type code for \(3 \le k \le n-4\).
\end{proof}

\section{Conclusion}\label{Conclusion}

In this paper, we mainly investigate some properties of  $\mathcal{C}$ codes. 
In \Cref{Th 1}, we prove  that $\mathcal{C}$ is the extended code of $\mathcal{C}_1$.
The necessary and sufficient conditions for $\mathcal{C}$ to be MDS and AMDS are presented in \Cref{Th 2,Th 3}, respectively. 
We also characterize the AMDS property of their dual codes in \Cref{Th 4} and show that  $\mathcal{C}$ is NMDS. 
Furthermore, we verify in \Cref{Th 7} that the extended code $\mathcal{C}$ is of  the non-GRS property.
In future work, we will investigate the covering radii and deep holes of these codes.

\section{Acknowledgment and Declarations}
\begin{itemize}
  \item {\bfseries Funding:}
  The paper is supported by the Fundamental Research Funds for the Central Universities (26CX03010A), China Education Innovation Research Fund (2022BL027), Shandong Provincial Natural Science Foundation of China (ZR2023LLZ013) and the Key Project of Computing Power Internet and Information Security, Ministry of Education (2024ZD006).
  \item {\bfseries Conflict of interest:} 
  The authors declare no conflict of interest.
\end{itemize}

\bibliographystyle{plain}
\bibliography{ref.bib}

\end{document}